\documentclass[11pt,onecolumn,twoside]{IEEEtran}%
\usepackage{amsfonts}%
\usepackage{amsmath}%
\usepackage{amssymb}%
\usepackage{graphicx}
\usepackage{subfigure}
\usepackage{setspace}
\usepackage{multirow}
\providecommand{\U}[1]{\protect\rule{.1in}{.1in}}

\begin{document}
\doublespacing
\title{Guarantees of Augmented Trace Norm Models in Tensor Recovery}
\author{Ziqiang Shi,~Jiqing Han~\IEEEmembership{Member IEEE},~Tieran Zheng,~Shiwen Deng,~Ji Li\thanks{Manuscript received XX, XX, 2011; revised XX, XX, 2011. This research was partly supported by the National Natural Science Foundation of China under grant No. 91120303 and No. 61071181.}\thanks{Ziqiang Shi, Jiqing Han, Tieran Zheng, Shiwen Deng, and JiLi
are all with the School of Computer Science and Technology,
Harbin Institute of Technology, No. 92, West Da-Zhi Street, Harbin, Heilongjiang,
China (Tel: +86-451-86417981, Email: \{zqshi, jqhan, zhengtieran,liji\}@hit.edu.cn), dengswen@126.com.}}
\maketitle

\begin{abstract}
This paper studies the recovery guarantees of the models of minimizing $\|\mathcal{X}\|_*+\frac{1}{2\alpha}\|\mathcal{X}\|_F^2$ where $\mathcal{X}$ is a tensor and $\|\mathcal{X}\|_*$ and $\|\mathcal{X}\|_F$ are the trace and Frobenius norm of respectively. We show that they can efficiently recover low-rank tensors.
In particular, they enjoy exact guarantees similar to those known for minimizing
$\|\mathcal{X}\|_*$ under the conditions on the sensing operator such as its null-space property, restricted
isometry property, or spherical section property. To recover a low-rank tensor
$\mathcal{X}^0$, minimizing $\|\mathcal{X}\|_*+\frac{1}{2\alpha}\|\mathcal{X}\|_F^2$ returns the same solution as minimizing $\|\mathcal{X}\|_*$ almost whenever
$\alpha\geq10\mathop {\max}\limits_{i}\|X^0_{(i)}\|_2$.
\end{abstract}


\section{Introduction}

Low-rank tensor recover problem is the generalization of sparse vector recovery and low-rank matrix recover to tensor data~\cite{gandy2011tensor,tomioka2010estimation,liu2009tensor}. It has drawn lots of attention from researchers
in different fields in the past several years. They have wide applications in data-mining, computer vision, signal/image
processing, machine learning, etc.. The fundamental problem of low-rank
tensor recovery is to find a tensor of (nearly) lowest rank from an underdetermined $\mathfrak{F}(\mathcal{X})=b$, where $\mathfrak{F}$ is
a linear operator and $\mathcal{X}\in \mathbb{R}^{I_1\times I_2\times \cdots \times I_N}$ is a $N$-way tensor.

To recover a low-rank tensor $\mathcal{X}^0\in \mathbb{R}^{I_1\times I_2\times \cdots \times I_N}$ from linear measurements $b=\mathfrak{F}(\mathcal{X}^0)$, a powerful approach is the convex model~\cite{gandy2011tensor,liu2009tensor}
\begin{equation}\label{eq:ConvexModel}
\mathop {\min}\limits_{\mathcal{X}}\{\|\mathcal{X}\|_*:\mathfrak{F}(\mathcal{X})=b\},
\end{equation}
where
\begin{equation}\label{eq:TensorTraceNormDefinition}
\|\mathcal{X}\|_*:=\frac{1}{N} \sum\limits_{i=1}^N \|X_{(i)}\|_*
\end{equation}
and $X_{(i)}$ is the mode-$i$ unfolding of $\mathcal{X}$, $\|X_{(i)}\|_*$ is the trace norm of the matrix $X_{(i)}$, i.e. the sum of the singular values of $X_{(i)}$. For vector $b$ with noise or generated by an approximately low-rank tensor, a variant of~(\ref{eq:ConvexModel}) is~\cite{liu2009tensor}
\begin{equation}\label{eq:ConvexModelVariant}
\mathop {\min}\limits_{\mathcal{X}}\{\|\mathcal{X}\|_*:\|\mathfrak{F}(\mathcal{X})-b\|_2\leq\sigma\}.
\end{equation}
Despite empirical success, the recovery guarantees of tensor recovery algorithms has not
been fully elucidated. Recently, several authors~\cite{lai2012augmented,recht2007guaranteed,mo2011new} have got excellent results in the guarantees of sparse vector recovery and low-rank matrix recover. In this paper, we try to generalize these results to low-rank tensor recovery. To the best of our knowledge, this is the first paper that studies the guarantees
of low-rank tensor recovery algorithm.

This paper mainly studies the guarantees of minimization of the augmented objective $\|\mathcal{X}\|_*+\frac{1}{2\alpha}\|\mathcal{X}\|_F^2$. The augmented model for~(\ref{eq:ConvexModel}) and~(\ref{eq:ConvexModelVariant}) are
\begin{equation}\label{eq:AugModel}
\mathop {\min}\limits_{\mathcal{X}}\{\|\mathcal{X}\|_*+\frac{1}{2\alpha}\|\mathcal{X}\|_F^2:\mathfrak{F}(\mathcal{X})=b\}.
\end{equation}
and
\begin{equation}\label{eq:AugModelVariant}
\mathop {\min}\limits_{\mathcal{X}}\{\|\mathcal{X}\|_*+\frac{1}{2\alpha}\|\mathcal{X}\|_F^2:\|\mathfrak{F}(\mathcal{X})-b\|_2\leq\sigma\}.
\end{equation}
respectively. These are natural generalizations of the augmented model for vector and matrix data~\cite{lai2012augmented} to tensor case.


\newtheorem{thm}{Theorem}
\newtheorem{cor}[thm]{Corollary}
\newtheorem{lem}[thm]{Lemma}
\newtheorem{mydef}{Definition}
\newtheorem{myRem}{Remark}

\section{Notations}
\label{sec:Notations}

We adopt the nomenclature mainly used by Kolda and Bader on tensor decompositions and applications~\cite{kolda2009tensor}, and also a few symbols of~\cite{de2000multilinear,tomioka2011statistical}.

The \emph{order} $N$ of a tensor is the number of dimensions, also known as ways or modes. Matrices (tensor of order two) are denoted by upper case letters, e.g. $X$, and lower case letters for the elements, e.g. $x_{ij}$. Higher-order tensors (order three or higher) are denoted by Euler script letters, e.g. $\mathcal{X}$, and element $(i_1,i_2,\cdots,i_N)$ of a $N$-order tensor $\mathcal{X}$ is denoted by $x_{i_1i_2\cdots i_N}$. \emph{Fibers} are the higher-order analogue of matrix rows and columns. A fiber is defined by fixing every index but one. The mode-$n$ fibers are all vectors $x_{i_1\cdots i_{n-1}:i_{n+1}\cdots i_N}$ that obtained by fixing the values of $\{i_1, i_2,\cdots, i_N\} \setminus i_n$. The mode-$n$ \emph{unfolding}, also knows as \emph{matricization}, of a tensor $\mathcal{X}\in \mathbb{R}^{I_1\times I_2\times \cdots \times I_N}$ is denoted by $X_{(n)}$ and arranges the model-$n$ fibers to be the columns of the resulting matrix. The unfolding operator is denoted as $\text{unfold}(\cdot)$. The opposite operation is $\text{refold}(\cdot)$, denotes the refolding of the matrix into a tensor.  The tensor element $(i_1,i_2,\cdots,i_N)$ is mapped to the matrix element $(i_n,j)$, where
\begin{equation*}
j=1+\displaystyle{\sum\limits_{\begin{subarray}{|}k=1\\ k\neq n\end{subarray}}^N (i_k-1)J_k}\quad  \textrm{ with  }\quad   J_k=\prod\limits_{\begin{subarray}{|}m=1 \\ m\neq n\end{subarray} } ^{k-1} I_m
\end{equation*}
Therefore, $X_{(n)}\in \mathbb{R}^{I_n\times I_1\cdots I_{n-1}I_{n+1}\cdots I_N}$. The $n$-\emph{rank} of a $N$-dimensional tensor $\mathcal{X}$, denoted as $\text{rank}_n(\mathcal{X})$ is the column rank of $X_{(n)}$, i.e. the dimension of the vector space spanned by the mode-$n$ fibers. We say
a tensor $\mathcal{X}$ is rank $(r_1, ... ,r_N)$ when $r_k=\text{rank}_k(\mathcal{X})$ for $k=1, ... ,N$, and denoted as $\text{rank}(\mathcal{X})$. We introduce an ordering among tensors by the rank:$\text{rank}(\mathcal{X})\leq\text{rank}(\mathcal{Y})\Leftrightarrow (\text{rank}_1(\mathcal{X}), ... ,\text{rank}_N(\mathcal{X}))\preceq(\text{rank}_1(\mathcal{Y}), ... ,\text{rank}_N(\mathcal{Y})) \Leftrightarrow \text{rank}_i(\mathcal{X})\leq\text{rank}_i(\mathcal{Y}),i=1,...,N$.
The inner product of two same-size tensors $\mathcal{X},\mathcal{Y}\in \mathbb{R}^{I_1\times I_2\times \cdots \times I_N}$ is defined as $\text{vec}(\mathcal{X})^{\top}\cdot\text{vec}(\mathcal{Y})$, where
vec is a vectorization.
The corresponding norm is $\|\mathcal{X}\|_F=\sqrt{<\mathcal{X},\mathcal{X}>}$, which is often called the Frobenius norm.

The $n$-th mode product of a tensor $\mathcal{X}\in \mathbb{R}^{I_1\times I_2\times \cdots \times I_N}$ with a matrix $U\in\mathbb{R}^{J\times I_n}$ is denoted by $\mathcal{X}\times_nU$ and is of size $I_1\times ... \times I_{n-1}\times J \times  I_{n+1} \times  \cdots \times I_N$. Elementwise, we have
\begin{equation}\label{eq:ModeNProductElementwise}
(\mathcal{X}\times_nU)_{i_1...i_{n-1}ji_{n+1}...i_N}=\sum\limits_{i_n=1}^{I_n}x_{i_1i_2...i_N}U_{ji_n}.
\end{equation}
Every tensor can be written as the product~\cite{de2000multilinear}
 \begin{equation}\label{eq:HOSVD}
\mathcal{A}=\mathcal{S}\times_1U^{(1)}\times_2U^{(2)}...\times_NU^{(N)},
\end{equation}
in which:
\begin{itemize}
  \item $U^{(n)}$ is a unitary $I_n\times I_n$ matrix,
  \item $\mathcal{S}$ is a $I_1\times I_2\times \cdots \times I_N$-tensor of which the subtensors $\mathcal{S}_{i_n=\alpha}$,
obtained by fixing the $n$-th index to $\alpha$, have the properties of:

1) all-orthogonality: two subtensors $\mathcal{S}_{i_n=\alpha}$ and $\mathcal{S}_{i_n=\beta}$ are orthogonal for all possible values of $n$, $\alpha$ and $\beta$ subject to $\alpha\neq\beta$: $<\mathcal{S}_{i_n=\alpha},\mathcal{S}_{i_n=\beta}>=0$,

2) ordering: for all possible values of $n$, one has $\|\mathcal{S}_{i_n=1}\|_F\geq\|\mathcal{S}_{i_n=2}\|_F\geq...\geq\|\mathcal{S}_{i_n=I_n}\|_F\geq0$.

The Frobenius norms $\|\mathcal{S}_{i_n=i}\|_F$, symbolized by $\sigma_i^{(n)}$, are mode-$n$ singular values of $\mathcal{A}$, that means the singular values of $A_{(n)}$.
\end{itemize}
This is called the higher-order singular
value decomposition (HOSVD) of a tensor $\mathcal{A}$ in~\cite{de2000multilinear}. Some properties of this HOSVD which will be used in this paper are list below as lemmas:

\begin{lem}
\label{lem:lemmaHOSVD0}
(\cite{de2000multilinear} Property 6). Let the HOSVD of $\mathcal{A}$ be given as in~(\ref{eq:HOSVD}), and let $r_n$ be equal to the highest index for which $\|\mathcal{S}_{i_n=r_n}\|_F>0$; then one has
$\text{rank}(A_{(n)})=r_n$.
\end{lem}

\begin{lem}
\label{lem:lemmaHOSVD1}
(~\cite{de2000multilinear} Property 8). Let the HOSVD of $\mathcal{A}$ be given as in~(\ref{eq:HOSVD}); due to the unitarily
invariant of the Frobenius norm, one has
$\|\mathcal{A}\|_F^2=\sum\limits_{i=1}^{I_1}\left (\sigma_i^{(1)}\right )^2=...=\sum\limits_{i=1}^{I_N}\left (\sigma_i^{(N)}\right )^2=\|\mathcal{S}\|_F^2$.
\end{lem}

\section{Motivations and contributions}

To explain why model~(\ref{eq:AugModel}) is interesting, we conducted following tensor completion simulations
\begin{equation}\label{eq:AugModelTest}
\mathop {\min}\limits_{\mathcal{X}}\{\|\mathcal{X}\|_*+\frac{1}{2\alpha}\|\mathcal{X}\|_F^2:x_{i_1i_2...i_N}=m_{i_1i_2...i_N},{i_1i_2...i_N}\in \Omega\}.
\end{equation}
to compare it with model~(\ref{eq:ConvexModel}) based tensor completion
\begin{equation}\label{eq:ConvexModelTest}
\mathop {\min}\limits_{\mathcal{X}}\{\|\mathcal{X}\|_*:x_{i_1i_2...i_N}=m_{i_1i_2...i_N},{i_1i_2...i_N}\in \Omega\}.
\end{equation}
The facade image data of~\cite{liu2009tensor} was used here to be an example.
Models~(\ref{eq:AugModelTest}) and~(\ref{eq:ConvexModelTest}) were solved to high accuracy by the solver
LRTC~\cite{liu2009tensor}. For each model, we measured and recorded
\begin{equation}\label{eq:RelativeErr}
\text{recovery relative error}: \|\mathcal{X}^*-\mathcal{X}_0\|_F/\|\mathcal{X}_0\|_F.
\end{equation}
The relative errors are depicted as functions of the number of iterations in Figure 1(a).

\begin{figure}
\subfigure[Recovery relative errors]{
\includegraphics[width=0.25\textwidth]{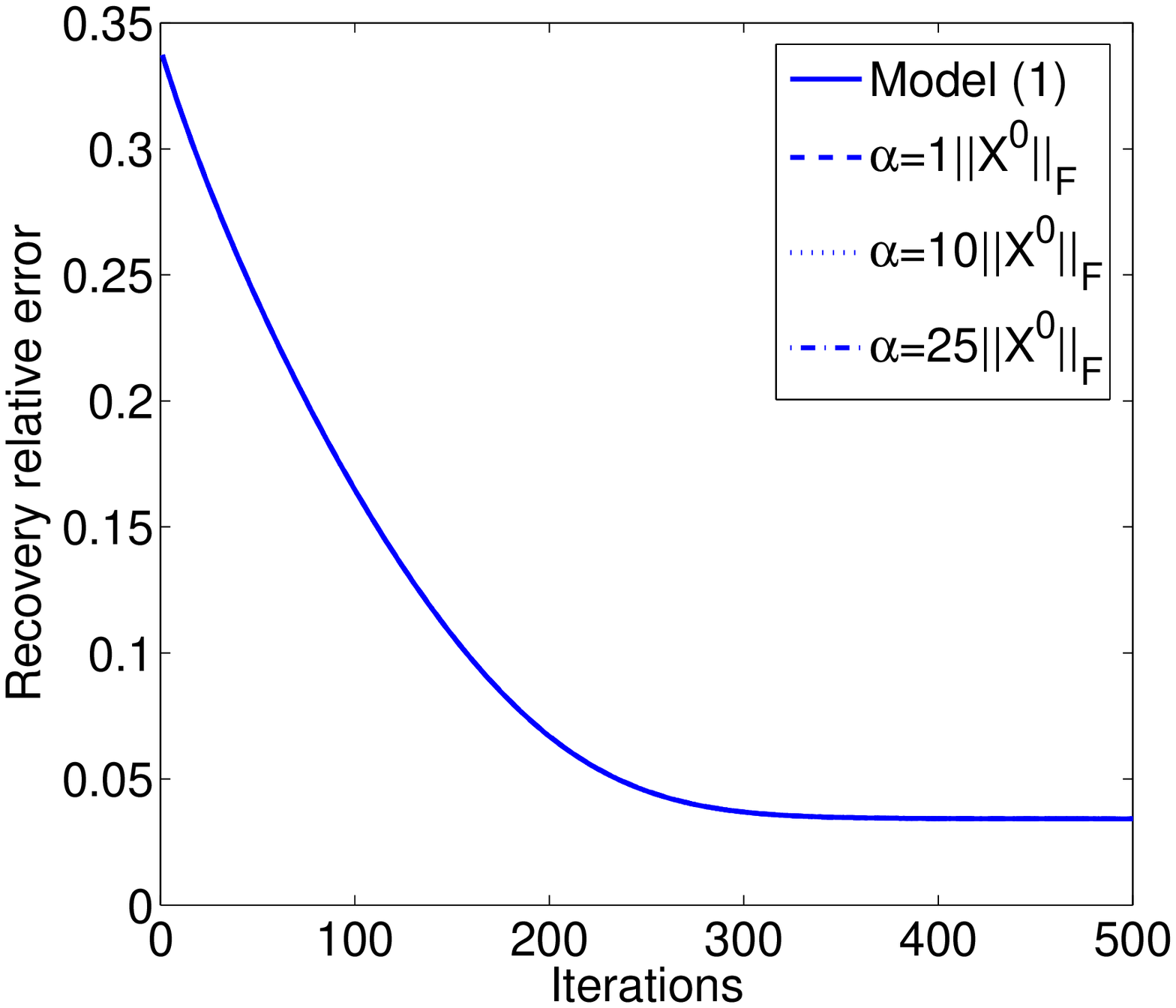}
}
\subfigure[Original image]{
\includegraphics[width=0.31\textwidth]{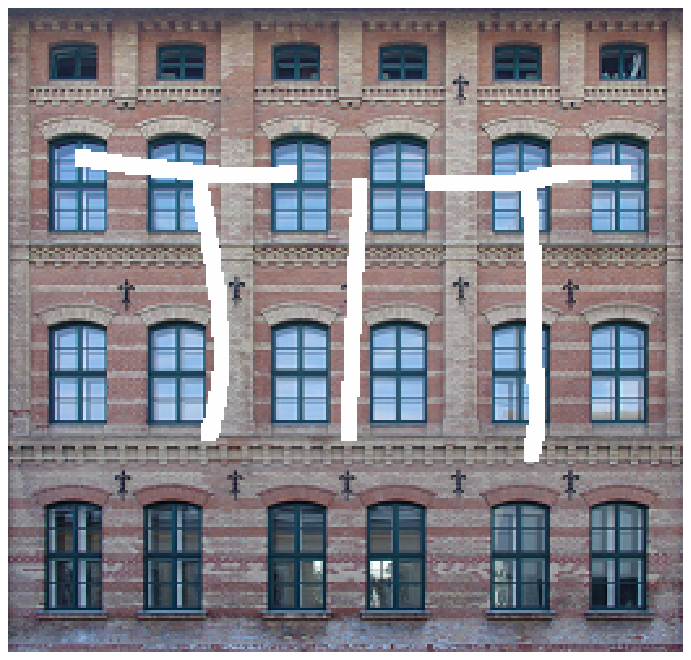}
}
\subfigure[Image recovered via~(\ref{eq:ConvexModel})]{
\includegraphics[width=0.31\textwidth]{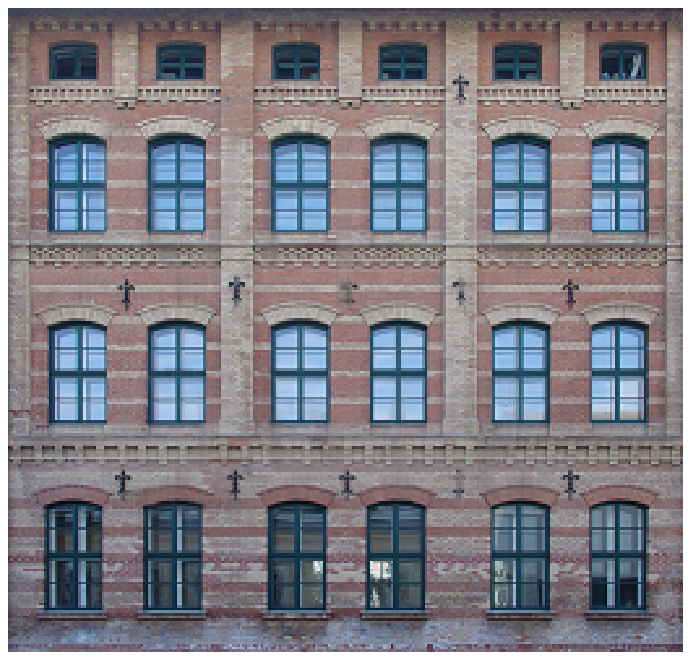}
}
\subfigure[Image recovered via~(\ref{eq:AugModel}) with $\alpha=1||\mathcal{X}^0||_F$]{
\includegraphics[width=0.31\textwidth]{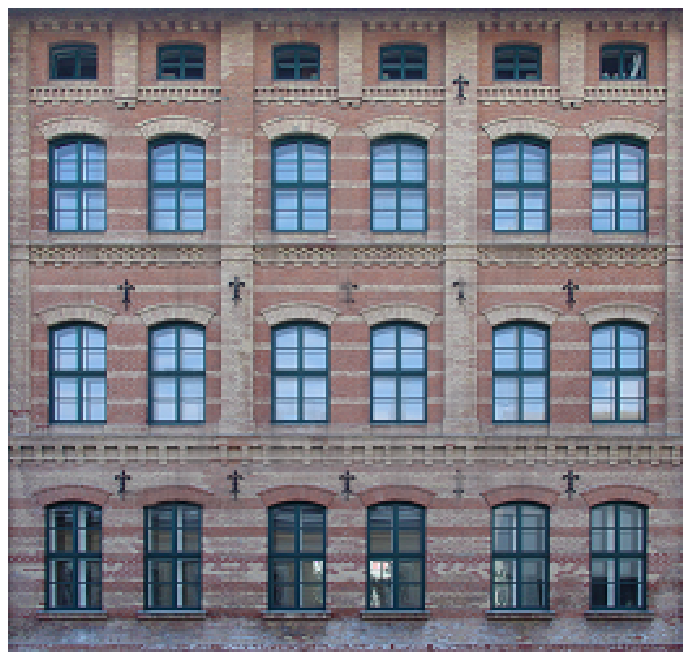}
}
\subfigure[Image recovered via~(\ref{eq:AugModel}) with $\alpha=10||\mathcal{X}^0||_F$]{
\includegraphics[width=0.31\textwidth]{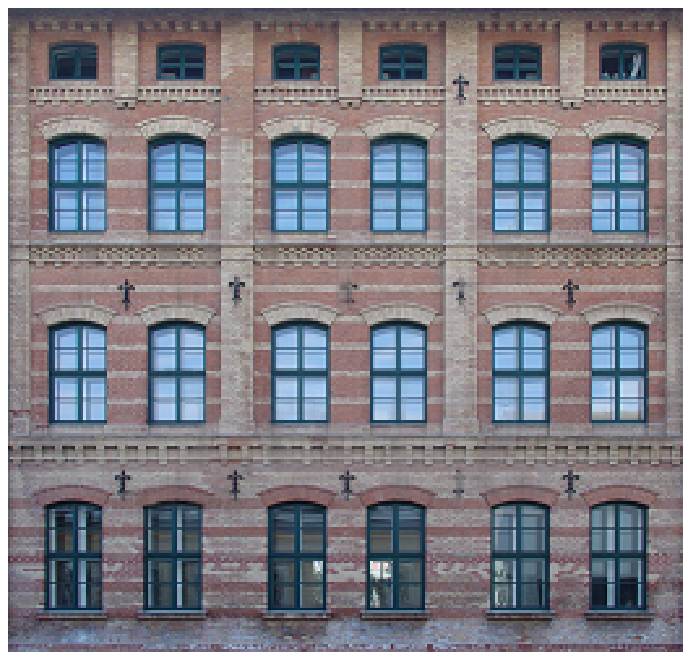}
}
\subfigure[Image recovered via~(\ref{eq:AugModel}) with $\alpha=25||\mathcal{X}^0||_F$]{
\includegraphics[width=0.31\textwidth]{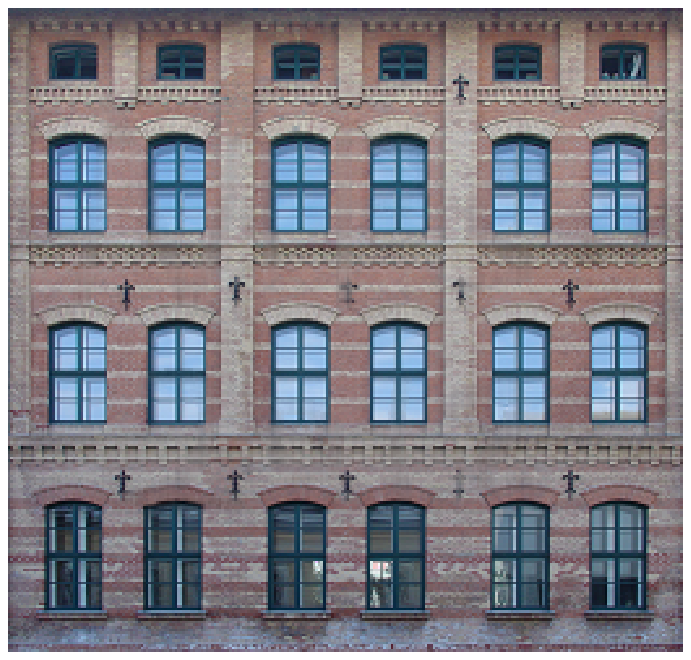}
}
\caption{Facade in-painting.}
\label{fig:PerformancComparOnline}       
\end{figure}

Motivated by the above example, we show in this paper that any $\alpha\geq10\mathop {\max}\limits_{i}\|X^0_{(i)}\|_2$ guarantees that problem~(\ref{eq:AugModel}) either recovers $\mathcal{X}^0$ exactly or returns an approximate of it nearly as good as the solution of problem~(\ref{eq:ConvexModel}). Specifically, we show that several properties of $\mathfrak{F}$, such as the null-space property (a simple condition
used in, e.g.,~\cite{liu2009tensor,donoho2001uncertainty,gribonval2003sparse,zhang2005simple,oymak2011simplified}), the restricted isometry principle~\cite{candes2005decoding}, and the spherical section property~\cite{zhang2008theory}, which have been used in the recovery guarantees for vectors and matrices, can also guarantee
the tensor recovery by model~(\ref{eq:ConvexModel}) and~(\ref{eq:AugModel}).

Even though $\mathcal{X}^0$ not known when $\alpha$ is set, $\mathop {\max}\limits_{i}\|X^0_{(i)}\|_2$ is often easy to estimate.
When $\mathop {\max}\limits_{i}\|X^0_{(i)}\|_2$ is not available, using inequalities $\|\mathcal{X}^0\|_F\geq\mathop {\max}\limits_{i}\|X^0_{(i)}\|_2$, one get
the more conservative formulae $\alpha\geq10\|\mathcal{X}^0\|_F$. Furthermore, when $\mathfrak{F}$
satisfies the RIP, one has $\|\mathcal{X}^0\|_F\leq\beta\|b\|_2$ for some $\beta>0$; hence, one has the option to
use the even more conservative formula $\alpha\geq10\beta\|b\|_2$.

\section{Tensor recovery guarantees}

This section establishes recovery guarantees for the original and augmented trace norm models~(\ref{eq:ConvexModel}) and~(\ref{eq:AugModel}). The results are given based on the properties of $\mathfrak{F}$ including the null-space property (NSP) in Theorem~\ref{thm:NSPcondition} and~\ref{thm:NSPconditionAug}, the restricted isometry principle (RIP)~\cite{candes2005decoding} in Theorem~\ref{thm:RIPconditionFor1} and~\ref{thm:RIPconditionFor2}, the spherical section property (SSP)~\cite{zhang2008theory} in Theorem~\ref{thm:SSPconditionFor1}. These results adapt and generalize of the work in~\cite{lai2012augmented}.

\subsection{Null space property}
The wide use of NSP for recovering sparse vector and low-rank matrices can be found in e.g.~\cite{donoho2001uncertainty,gribonval2003sparse,zhang2005simple,oymak2011simplified}. In this subsection, we extend the NSP conditions on $\mathfrak{F}$ for tensor recovery.
 Throughout this subsection, we let $\sigma_i(X), i=1,...,m$ denote the $i$-th largest singular value of matrix $X$ of rank $m$ or less, and $\Sigma(X)=\text{diag}(\sigma_1(X), ..., \sigma_m(X))$ denote the diagonal matrix of singular values and $s(X)=(\sigma_1(X), ..., \sigma_m(X))$. $\|X\|_2=\sigma_1(X)$ denotes the spectral norms of $X$.

 We will need the following two technical lemmas for the introduction of the tensor NSP conditions.
\begin{lem}
\label{lem:Theorem7.4.51}
([8] Theorem 7.4.51). Let $X$ and $Y$ be two matrices of the same size. Then we have
\begin{equation}\label{eq:1normMatrixInequality}
\sum\limits_{i=1}^{m}|\sigma_i(X)-\sigma_i(Y)|\leq \|X-Y\|_*
\end{equation}
and
\begin{equation}\label{eq:2normMatrixInequality}
\sum\limits_{i=1}^{m}(\sigma_i(X)-\sigma_i(Y))^2\leq \|X-Y\|_F^2.
\end{equation}
\end{lem}

\begin{lem}
\label{lem:1normAnd2normInequality}
([4] Equation (19)). Let $x$ and $h$ be two vectors of the same size, $\mathcal{S}:=\text{supp}(x)$ and $\mathcal{Z}=\mathcal{S}^c$. Then we have
\begin{equation}\label{eq:1normVectorInequality}
\|x+h\|_1\geq\|x\|_1+\|h_{\mathcal{Z}}\|_1-\|h_{\mathcal{S}}\|_1
\end{equation}
and
\begin{equation}\label{eq:2normVectorInequality}
\|x+h\|_1+\frac{1}{2\alpha}\|x+h\|_2^2\geq\left [\|x\|_1+\frac{1}{2\alpha}\|x\|_2^2\right ]+\left [\|h_{\mathcal{Z}}\|_1-(1+\frac{\|x_{\mathcal{S}}\|_\infty}{\alpha})\|h_{\mathcal{S}}\|_1\right ]+\frac{1}{2\alpha}\|h\|_2^2.
\end{equation}
\end{lem}

Now we give a NSP type sufficient condition for problem~(\ref{eq:ConvexModel}).

\begin{thm}
\label{thm:NSPcondition}
(Tensor NSP condition for~(\ref{eq:ConvexModel})). Assume $\mathcal{X}^0$ is fixed, problem~(\ref{eq:ConvexModel}) uniquely recovers all tensors $\mathcal{X}^0$ of rank $(r_1, ... ,r_N)$ or less from the measurements $\mathfrak{F}(\mathcal{X}^0)=b$, if all $\mathcal{H}\in \text{Null}(\mathfrak{F})\backslash\{0\}$ satisfy
\begin{equation}\label{eq:NSPcondition}
\sum\limits_{i=1}^N\sum\limits_{j=r_i+1}^{I_i}\sigma_j(H_{(i)})>\sum\limits_{i=1}^N\sum\limits_{j=1}^{r_i}\sigma_j(H_{(i)}).
\end{equation}
\end{thm}

\begin{proof}
Pick any tensor $\mathcal{X}^0$ of rank $(r_1, ... ,r_N)$ or less and let $b=\mathfrak{F}(\mathcal{X}^0)$. For any $\mathcal{H}\in \text{Null}(\mathfrak{F})\backslash\{0\}$, we have $\mathfrak{F}(\mathcal{X}^0+\mathcal{H})=\mathfrak{F}(\mathcal{X}^0)=b$. By using~(\ref{eq:1normMatrixInequality}), we have
\begin{equation}\label{eq:ConvexModelSufficiencyDerivative}
\begin{split}
\|\mathcal{X}^0+\mathcal{H}\|_*=&\frac{1}{N}\sum\limits_{i=1}^N\|X^0_{(i)}+H_{(i)}\|_* \geq \frac{1}{N}\sum\limits_{i=1}^N\|s(X^0_{(i)})-s(H_{(i)})\|_1 \\
\geq &\frac{1}{N}\sum\limits_{i=1}^N(\|X^0_{(i)}\|_*+[\sum\limits_{j=r_i+1}^{I_i}\sigma_j(H_{(i)})-\sum\limits_{j=1}^{r_i}\sigma_j(H_{(i)})]) \\
=& \|\mathcal{X}^0\|_*+\frac{1}{N}\sum\limits_{i=1}^N(\sum\limits_{j=r_i+1}^{I_i}\sigma_j(H_{(i)})-\sum\limits_{j=1}^{r_i}\sigma_j(H_{(i)}))
\end{split}
\end{equation}
where the first inequality follows from~(\ref{eq:1normVectorInequality}). For any nonzero $\mathcal{H}\in \text{Null}(\mathfrak{F})$, $\|\mathcal{H}\|_F^2>0$.
Hence from~(\ref{eq:NSPcondition}) and~(\ref{eq:ConvexModelSufficiencyDerivative}), it follows that $\mathcal{X}^0$ is unique minimizer of~(\ref{eq:ConvexModel}).
\end{proof}

We can extend this result to problem~(\ref{eq:AugModel}) as follows.

\begin{thm}
\label{thm:NSPconditionAug}
(Tensor NSP condition for~(\ref{eq:AugModel})). Assume $\mathcal{X}^0$ is fixed, problem~(\ref{eq:AugModel}) uniquely recovers all tensors $\mathcal{X}^0$ of rank $(r_1, ... ,r_N)$ or less from the measurements $\mathfrak{F}(\mathcal{X}^0)=b$, if all $\mathcal{H}\in \text{Null}(\mathfrak{F})\backslash\{0\}$ satisfy
\begin{equation}\label{eq:NSPconditionAug}
\sum\limits_{i=1}^N\sum\limits_{j=r_i+1}^{I_i}\sigma_j(H_{(i)})\geq\sum\limits_{i=1}^N(1+\frac{\|X^0_{(i)}\|_2}{\alpha})\sum\limits_{j=1}^{r_i}\sigma_j(H_{(i)}).
\end{equation}
\end{thm}

\begin{proof}
Pick any tensor $\mathcal{X}^0$ of rank $(r_1, ... ,r_N)$ or less and let $b=\mathfrak{F}(\mathcal{X}^0)$. For any nonzero $\mathcal{H}\in \text{Null}(\mathfrak{F})$, we have $\mathfrak{F}(\mathcal{X}^0+\mathcal{H})=\mathfrak{F}(\mathcal{X}^0)=b$. Thus
\begin{equation}\label{eq:SufficiencyDerivative}
\begin{split}
&\|\mathcal{X}^0+\mathcal{H}\|_*+\frac{1}{2\alpha}\|\mathcal{X}^0+\mathcal{H}\|_F^2=\frac{1}{N}\sum\limits_{i=1}^N(\|X^0_{(i)}+H_{(i)}\|_*+\frac{1}{2\alpha}\|X^0_{(i)}+H_{(i)}\|_F^2) \\
\geq &\frac{1}{N}\sum\limits_{i=1}^N(\|s(X^0_{(i)})-s(H_{(i)})\|_1+\frac{1}{2\alpha}\|s(X^0_{(i)})-s(H_{(i)})\|_2^2)\\
\geq &\frac{1}{N}\sum\limits_{i=1}^N([\|X^0_{(i)}\|_*+\frac{1}{2\alpha}\|X^0_{(i)}\|_F^2]+[\sum\limits_{j=r_i+1}^{I_i}\sigma_j(H_{(i)})-(1+\frac{\|X^0_{(i)}\|_2}{\alpha})\sum\limits_{j=1}^{r_i}\sigma_j(H_{(i)})]+\frac{1}{2\alpha}\|H_{(i)}\|_F^2)\\
=& \|\mathcal{X}^0\|_*+\frac{1}{2\alpha}\|\mathcal{X}^0\|_F^2+\frac{1}{2\alpha}\|\mathcal{H}\|_F^2+\frac{1}{N}\sum\limits_{i=1}^N(\sum\limits_{j=r_i+1}^{I_i}\sigma_j(H_{(i)})-(1+\frac{\|X^0_{(i)}\|_2}{\alpha})\sum\limits_{j=1}^{r_i}\sigma_j(H_{(i)}))
\end{split}
\end{equation}
where the first inequality follows from~(\ref{eq:1normMatrixInequality}) and~(\ref{eq:2normMatrixInequality}), and the second inequality follows from~(\ref{eq:1normVectorInequality}) and~(\ref{eq:2normVectorInequality}). For any nonzero $\mathcal{H}\in \text{Null}(\mathfrak{F})$, $\|\mathcal{H}\|_F^2>0.$ Hence, from~(\ref{eq:SufficiencyDerivative}) and~(\ref{eq:NSPconditionAug}), it follows that $\mathcal{X}^0+\mathcal{H}$ leads to a strictly worse objective than $\mathcal{X}^0$. That is, $\mathcal{X}^0$ is the unique solution to problem~(\ref{eq:AugModel}).

\end{proof}

\begin{myRem}
For any finite $\alpha>0$,~(\ref{eq:NSPconditionAug}) is stronger than~(\ref{eq:NSPcondition}) due to the extra term $\frac{\|X^0_{(i)}\|_2}{\alpha}$. Since various
uniform recovery results establish conditions that guarantee~(\ref{eq:NSPcondition}), one can tighten these conditions so that
they guarantee~(\ref{eq:NSPconditionAug}) and thus the uniform recovery by problem~(\ref{eq:AugModel}). How much tighter these conditions have
to be depends on the value $\frac{\|X^0_{(i)}\|_2}{\alpha}$.
\end{myRem}

\subsection{Tensor restricted isometry principle}

In this subsection, we generalize the RIP-based guarantees to tensor case and show that any $\alpha\geq10\mathop {\max}\limits_{i}\|X^0_{(i)}\|_2$ guarantees exact recovery by~(\ref{eq:AugModel}).
\begin{mydef}
(Tensor RIP). Let $\mathfrak{M}_{(r_1,r_2,...,r_N)}:=\{\mathcal{X}\in \mathbb{R}^{I_1\times I_2\times \cdots \times I_N}:\text{rank}(X_{(n)})\leq r_n, n=1,...,N$\}. The RIP constant $\delta_{(r_1,r_2,...,r_N)}$ of linear operator $\mathfrak{F}$ is the smallest value such that
\begin{equation}\label{eq:TensorRIP}
(1-\delta_{(r_1,r_2,...,r_N)})\|\mathcal{X}\|_F^2\leq\|\mathfrak{F}(\mathcal{X})\|_2^2\leq(1+\delta_{(r_1,r_2,...,r_N)})\|\mathcal{X}\|_F^2
\end{equation}
holds for all $\mathcal{X}\in\mathfrak{M}_{(r_1,r_2,...,r_N)}$.
\end{mydef}

The following recovery theorems will characterize the power of the tensor restricted isometry
constants. The first theorem generalizes Lemma 1.3 in~\cite{candes2005decoding} and Theorem 3.2 in\cite{recht2007guaranteed} to low-rank tensor recovery.

\begin{thm}
\label{thm:Unique}
Suppose $\delta_{(2r_1,2r_2,...,2r_N)} < 1$ for some $(r_1,r_2,...,r_N)\succeq(1,1,...,1)$. Then $\mathcal{X}^0$ is the only tensor of rank at most $(r_1,r_2,...,r_N)$ satisfying $\mathfrak{F}(\mathcal{X})=b$.
\end{thm}

\begin{proof}
Assume, on the contrary, that there exists a tensor with rank $(r_1,r_2,...,r_N)$ or less satisfying $\mathfrak{F}(\mathcal{X})=b$ and $\mathcal{X}\neq\mathcal{X}_0$. Then $\mathcal{Z}:=\mathcal{X}-\mathcal{X}_0$ is a nonzero tensor of rank at most $(2r_1,2r_2,...,2r_N)$, and $\mathfrak{F}(\mathcal{Z})=0$. But then we would have $0=\|\mathfrak{F}(\mathcal{Z})\|_2^2\geq(1-\delta_{(2r_1,2r_2,...,2r_N)})\|\mathcal{Z}\|_F^2>0$ which is a contradiction.
\end{proof}


The proof of the preceding theorem is identical to the argument given by Candes and Tao and
is an immediate consequence of our definition of the constant $\delta_{(r_1,r_2,...,r_N)}$. No adjustment is necessary in the
transition from sparse vectors and low-rank matrices to low-rank tensors. The key property used is the sub-additivity of
the rank. Adapting results in~\cite{lai2012augmented,mo2011new}, we give the uniform recovery conditions for~(\ref{eq:ConvexModel}) below.

\begin{thm}
\label{thm:RIPconditionFor1}
(RIP condition for exact recovery by~(\ref{eq:ConvexModel})).
Let $\mathcal{X}^0$ be a tensor with rank $(r_1,r_2,...,r_N)$ or less. Problem~(\ref{eq:ConvexModel}) exactly recovers $\mathcal{X}^0$ from measurements $b=\mathfrak{F}(\mathcal{X}^0)$ if $\mathfrak{F}$ satisfies the RIP with $\delta_{(I_1,...,2r_n,...,I_N)}<0.4931$, for $n=1,...N$.
\end{thm}

\begin{proof}
For any nonzero $\mathcal{H}\in \text{Null}(\mathfrak{F})$, $\|\mathcal{H}\|_F>0$, let $\mathcal{H}=\mathcal{S}\times_1 U^{(1)}\times_2 U^{(2)}...\times_N U^{(N)}$ be the HOSVD of $\mathcal{H}$.
 From Proposition 3.7 of~\cite{kolda2006multilinear} we have $\text{rank}(H_{(n)})\leq\text{rank}(S_{(n)}), (n=1,..,N)$, thus we have $\text{rank}(\mathcal{H})\preceq\text{rank}(\mathcal{S})$.
We decompose $\mathcal{H}=\mathcal{H}_0+\mathcal{H}_1+...$, where
\begin{equation}\label{eq:NormEquation}
\begin{split}
\mathcal{H}_0=\mathcal{S}_{\{i_n:1,...,r_n\}}\times_1 U^{(1)}\times_2 U^{(2)}...\times_N U^{(N)},   \\
\mathcal{H}_1=\mathcal{S}_{\{i_n:r_n+1,...,2r_n\}}\times_1 U^{(1)}\times_2 U^{(2)}...\times_N U^{(N)},...,
\end{split}
\end{equation}
and $\mathcal{S}_{\{i_n:\alpha\}}$ is the tensor obtained by fixing the n-th index to the index set $\alpha$, others to zero. Similarly we have $\text{rank}(\mathcal{H})_i\preceq\text{rank}(\mathcal{S}_{\{i_n:ir_n+1,...,(i+1)r_n\}})\preceq(I_1,...,r_n,...,I_n),(i=0,1,...)$.
Due to the unitarily
invariant of the Frobenius norm, we have $\|\mathcal{H}_0\|_F=\|\mathcal{S}_{\{i_n:1,...,r_n\}}\|_F$, $\|\mathcal{H}_1\|_F=\|\mathcal{S}_{\{i_n:r_n+1,...,2r_n\}}\|_F$, .... Let $h=(\sigma_{1}^{(n)},\sigma_{2}^{(n)},...,\sigma_{I_n}^{(n)})$ and $h_0=(\sigma_{1}^{(n)},...,\sigma_{r_n}^{(n)})$, $h_1=(\sigma_{r_n+1}^{(n)},...,\sigma_{2r_n}^{(n)})$, $h_2=(\sigma_{2r_n+1}^{(n)},...,\sigma_{3r_n}^{(n)})$,..., where $\sigma_{i}^{(n)}$ is the $i$-th largest mode-$n$ singular value.
From the definition of HOSVD and Lemma 2, we have
\begin{equation}\label{eq:NormEquation2}
\|\mathcal{H}_i\|_F^2=\|\mathcal{S}_{\{i_n:ir_n+1,...,(i+1)r_n\}}\|_F^2=\sum\limits_{j=ir_n+1}^{(i+1)r_n}\left (\sigma_j^{(n)}\right )^2=\|h_i\|_2^2, (i=0,1,...).
\end{equation}
Due to the mean-inequation, one has $\|h_i\|_2^2 \geq \|h_i\|_1^2/r_n, (i=0,1,...)$.
Assume that $\|h_1\|_1=t(\sum\nolimits_{i\geq1}\|h_i\|_1)$ with some $t\in[0,1]$. Then we have $(\sum\nolimits_{i\geq2}\|h_i\|_1)=(1-t)(\sum\nolimits_{i\geq1}\|h_i\|_1)$.

First from Lemma 2.1, Lemma 2.2 in~\cite{mo2011new} and~(\ref{eq:NormEquation}) we have
\begin{equation}\label{eq:LemaDerivative_1}
\sum\limits_{i\geq2}\|\mathcal{H}_i\|_F^2\leq \sigma_{2r_{n}+1}^{(n)}\sum\limits_{i\geq2}\|h_i\|_1\leq \frac{t}{r_n}(1-t)(\sum\limits_{i\geq1}\|h_i\|_1)^2
\end{equation}
and
\begin{equation}\label{eq:LemaDerivative0}
r_n^{\frac{1}{2}}\|\mathcal{H}_i\|_F\leq \|h_i\|_1+r_n(\sigma_{ir_n+1}^{(n)}-\sigma_{ir_{n+1}+r_{n+1}}^{(n)})/4, i=2,3,....
\end{equation}
From~(\ref{eq:LemaDerivative0}) we have
\begin{equation}\label{eq:LemaDerivative01}
r_n^{\frac{1}{2}}\sum\limits_{i\geq2}\|\mathcal{H}_i\|_F\leq \sum\limits_{i\geq2}\|h_i\|_1+r_n|d_{2r_n+1}|/4\leq  \sum\limits_{i\geq2}\|h_i\|_1+\|h_1\|_1/4=(1-3t/4)\sum\limits_{i\geq1}\|h_i\|_1.
\end{equation}
So we have
\begin{equation}\label{eq:LemaDerivative1}
\begin{split}
\|\mathfrak{F}(\mathcal{H}_0+\mathcal{H}_1))\|_2^2\geq(1-\delta_{(I_1,I_2,...,2r_n,...,I_N)})\|\mathcal{H}_0+\mathcal{H}_1\|_F^2 \\
=(1-\delta_{(I_1,I_2,...,2r_n,...,I_N)})\left (\|\mathcal{H}_0\|_F^2+\|\mathcal{H}_1\|_F^2 \right ) \\
\geq (1-\delta_{(I_1,I_2,...,2r_n,...,I_N)})\left (\|h_0\|_1^2+\|h_1\|_1^2 \right )/r_n \\
\end{split}
\end{equation}
and
\begin{equation}\label{eq:LemaDerivative2}
\begin{split}
&\|\mathfrak{F}(\mathcal{H}_2+\mathcal{H}_{3}+...))\|_2^2=\sum\limits_{j,k\geq2}<\mathfrak{F}(\mathcal{H}_{j}),\mathfrak{F}(\mathcal{H}_{k})>\\
&=\sum\limits_{j\geq2}<\mathfrak{F}(\mathcal{H}_{j}),\mathfrak{F}(\mathcal{H}_{j})>+2\sum\limits_{2\leq j< k}<\mathfrak{F}(\mathcal{H}_{j}),\mathfrak{F}(\mathcal{H}_{k})>\\
&\leq\sum\limits_{j\geq2}(1+\delta_{(I_1,I_2,...,r_n,...,I_N)})\|\mathcal{H}_{j}\|_F^2+2\delta_{(I_1,I_2,...,2r_n,...,I_N)}\sum\limits_{2\leq j< k}\|\mathcal{H}_{j}\|_F\|\mathcal{H}_{k}\|_F\\
&=\sum\limits_{j\geq2}\|\mathcal{H}_{j}\|_F^2+\delta_{(I_1,I_2,...,2r_n,...,I_N)}\left (\sum\limits_{j\geq2}\|\mathcal{H}_{j}\|_F \right )^2
\end{split}
\end{equation}
Further more, by~(\ref{eq:LemaDerivative_1}),(\ref{eq:LemaDerivative01}) we have
\begin{equation}\label{eq:LemaDerivative3}
\begin{split}
&\|\mathfrak{F}(\mathcal{H}_2+\mathcal{H}_{3}+...))\|_2^2\leq\sum\limits_{j\geq2}\|\mathcal{H}_{j}\|_F^2+\delta_{(I_1,I_2,...,2r_n,...,I_N)}\left (\sum\limits_{j\geq2}\|\mathcal{H}_{j}\|_F \right )^2
\\\leq&\frac{t}{r_n}(1-t)(\sum\limits_{i\geq1}\|h_i\|_1)^2+(1-3t/4)\sum\limits_{i\geq1}\|h_i\|_1
\\=&\frac{t(1-t)+\delta_{(I_1,I_2,...,2r_n,...,I_N)}(1-3t/4)^2}{r_n}(\sum\limits_{i\geq1}\|h_i\|_1)^2.
\end{split}
\end{equation}
Since $\mathfrak{F}(\mathcal{H}_0+\mathcal{H}_1+\mathcal{H}_2...)=\mathfrak{F}(\mathcal{H})=0$, we have $\|\mathfrak{F}(\mathcal{H}_0+\mathcal{H}_1))\|_2^2=\|\mathfrak{F}(\mathcal{H}_2+\mathcal{H}_{3}+...))\|_2^2$, by the above equations we have
\begin{equation}\label{eq:LemaDerivative4}
\begin{split}
(1-\delta_{(I_1,I_2,...,2r_n,...,I_N)})\left (\|h_0\|_1^2+\|h_1\|_1^2 \right )/r_n \\
\leq\frac{t(1-t)+\delta_{(I_1,I_2,...,2r_n,...,I_N)}(1-3t/4)^2}{r_n}(\sum\limits_{i\geq1}\|h_i\|_1)^2.
\end{split}
\end{equation}
Hence, let $\delta=\delta_{(I_1,I_2,...,2r_n,...,I_N)}$
\begin{equation}\label{eq:LemaDerivative5}
\|h_0\|_1^2 \leq\frac{1}{1-\delta}\left [ \delta +(1-3\delta/2)t-(2-25\delta/16)t^2\right](\sum\limits_{i\geq1}\|h_i\|_1)^2
\end{equation}
We have a quadratic polynomial of $t$ with $t\in[0,1]$ in the right-hand side of the above inequality. Hence, by calculus, this quadratic polynomial achieves its maximal value at $t=\frac{1-3\delta/2}{4-25\delta/8}\in[0,1]$. Therefore we obtain $\|h_0\|_1\leq\theta(\sum\limits_{i\geq1}\|h_i\|_1)$, where
\begin{equation}\label{eq:ThetaDef}
\theta=\sqrt{\frac{4(1+5\delta-4\delta^2)}{(1-\delta)(32-25\delta)}}.
\end{equation}
$\delta<(77-\sqrt{1337})/82\approx 0.4931$, then $\theta<1$, we get $\|h_0\|_1<(\sum\limits_{i\geq1}\|h_i\|_1)$, which is
\begin{equation}\label{eq:LemaDerivative6}
\sum\limits_{j=r_n+1}^{I_n}\sigma_j^{(n)}>\sum\limits_{j=1}^{r_n}\sigma_j^{(n)}.
\end{equation}
If for all $n=(1,...,N)$, we have~(\ref{eq:LemaDerivative6}), then we get~(\ref{eq:NSPcondition}).
\end{proof}

Next we carry out a similar study for the augmented model~(\ref{eq:AugModel}).
\begin{thm}
\label{thm:RIPconditionFor2}
(RIP condition for exact recovery by~(\ref{eq:AugModel})). Let $\mathcal{X}^0$ be a tensor with rank $(r_1,r_2,...,r_N)$ or less. The augmented model~(\ref{eq:ConvexModelVariant}) exactly recovers $\mathcal{X}^0$ from measurements $b=\mathfrak{F}(\mathcal{X}^0)$ if $\mathfrak{F}$ satisfies the RIP with $\delta_{(I_1,...,2r_n,...,I_N)}<0.4404, n=1,...N$ and $\alpha\geq10\mathop {\max}\limits_{i}\|X^0_{(i)}\|_2$.
\end{thm}

\begin{proof}
 The proof of Theorem~\ref{thm:RIPconditionFor1} establishes that any nonzero $\mathcal{H}\in \text{Null}(\mathfrak{F})$ satisfies $\|h_0\|_1\leq\theta(\sum\limits_{i\geq1}\|h_i\|_1)$. Hence,
if $(1+\frac{\|X^0_{(i)}\|_2}{\alpha})\theta\leq1$, notice $\theta<1$, we have
\begin{equation}\label{eq:LemaDerivative7}
     \alpha\geq(\theta^{-1}-1)^{-1} \|X^0_{(i)}\|_2=\frac{\|X^0_{(i)}\|_2\sqrt{4(1+5\delta-4\delta^2)}}{\sqrt{(1-\delta)(32-25\delta)}-\sqrt{4(1+5\delta-4\delta^2)}}.
\end{equation}
For $\delta=0.0.4404$, we obtain $(\theta^{-1}-1)^{-1} \|X^0_{(i)}\|_2\approx9.9849\|X^0_{(i)}\|_2\leq\alpha$, which proves the theorem.
\end{proof}

\begin{myRem}
Different values of $\delta_{(I_1,...,2r_n,...,I_N)}, n=1,...N$ are associated with different conditions on $\alpha$. Following~(\ref{eq:LemaDerivative7}), if $\delta_{(I_1,...,2r_n,...,I_N)}<0.4715, n=1,...N$, $\alpha\geq10\mathop {\max}\limits_{i}\|X^0_{(i)}\|_2$ guarantees exact recovery. If $\delta_{(I_1,...,2r_n,...,I_N)}<0.1273, n=1,...N$, $\alpha\geq10\mathop {\max}\limits_{i}\|X^0_{(i)}\|_2$ guarantees exact recovery. In
general, smaller $\delta_{(I_1,...,2r_n,...,I_N)}, n=1,...N$ allows a smaller $\alpha$.
\end{myRem}

\subsection{Spherical section property}
Next, we we derive exact conditions based on the spherical section property (SSP)~\cite{lai2012augmented,zhang2008theory}. There is not much discussion on spherical section property (SSP) for low-rank tensor recovery in the literature, here we present a SSP-based result.

\begin{thm}
\label{thm:SSPconditionFor1}
(SSP condition for exact recovery by~(\ref{eq:AugModel})). Let $\mathfrak{F}: \mathbb{R}^{I_1\times I_2\times \cdots \times I_N}\rightarrow \mathbb{R}^{m}$ be a linear operator. Suppose there exists $\triangle >0$ such that all nonzero nonzero $\mathcal{H}\in \text{Null}(\mathfrak{F})$ satisfy
\begin{equation}\label{eq:SSPcondition1}
    \frac{\|\mathcal{H}\|_*}{\|\mathcal{H}\|_F}\geq \sqrt{\frac{m}{\triangle}}.
\end{equation}
Assume that $\|X^0_{(i)}\|_2,(i=1,...,N)$ and $\alpha>0$ are fixed. If
 \begin{equation}\label{eq:SSPcondition2}
    m\geq(2+\frac{\|X^0_{(i)}\|_2}{\alpha})^2r_i\triangle, (i=1,...,N),
\end{equation}
then the null-space condition holds for all nonzero $\mathcal{H}\in \text{Null}(\mathfrak{F})$. Hence is sufficient for problem to recover any $\mathcal{X}^0$ with rank $(r_1,r_2,...,r_N)$ or less from measurements $b=\mathfrak{F}(\mathcal{X}^0)$.
\end{thm}

\begin{proof}
 Condition~(\ref{eq:NSPconditionAug}) is equivalent to
 \begin{equation}\label{eq:NSPconditionAug2}
\sum\limits_{i=1}^N\sum\limits_{j=1}^{I_i}\sigma_j(H_{(i)})\geq\sum\limits_{i=1}^N(2+\frac{\|X^0_{(i)}\|_2}{\alpha})\sum\limits_{j=1}^{r_i}\sigma_j(H_{(i)}).
\end{equation}
Since $\sum\limits_{j=1}^{r_i}\sigma_j(H_{(i)})\leq\sqrt{r_i}\sqrt{\sum\limits_{j=1}^{r_i}\sigma_j(H_{(i)})^2}\leq\sqrt{r_i}\|H_{(i)}\|_F$,~(\ref{eq:NSPconditionAug2}) holds provide that
\begin{equation}\label{eq:NSPconditionAug3}
\sum\limits_{i=1}^N\sum\limits_{j=1}^{I_i}\sigma_j(H_{(i)})\geq\sum\limits_{i=1}^N(2+\frac{\|X^0_{(i)}\|_2}{\alpha})\sqrt{r_i}\|H_{(i)}\|_F.
\end{equation}
Now from~(\ref{eq:SSPcondition1}) and~(\ref{eq:SSPcondition2}), one has
\begin{equation}\label{eq:SSPcondition3}
    \|\mathcal{H}\|_*\geq\mathop {\max}\limits_{i}(2+\frac{\|X^0_{(i)}\|_2}{\alpha})\sqrt{r_i}\|\mathcal{H}\|_F,
\end{equation}
which is equivalent to
\begin{equation}\label{eq:SSPcondition4}
    \frac{1}{N}\sum\limits_{i=1}^N\|H^0_{(i)}\|_*\geq\mathop {\max}\limits_{i}(2+\frac{\|X^0_{(i)}\|_2}{\alpha})\sqrt{r_i}\|\mathcal{H}\|_F \geq\frac{1}{N}\sum\limits_{i=1}^N(2+\frac{\|X^0_{(i)}\|_2}{\alpha})\sqrt{r_i}\|H_{(i)}\|_F.
\end{equation}
Thus the null-space condition holds.
\end{proof}

\section{Conclusion}

In this work we focussed on the recovery guarantees of tensor recovery via convex
optimization. We presented general results stating that the extension of
some sufficient conditions for the recovery of low-rank matrices
using nuclear norm minimization are also sufficient for the recovery of low-rank
tensors using tensor nuclear norm minimization. We extended the null-space property, the restricted isometry principle, and the spherical section property conditions to the augmented tensor recovery problems, and find that any $\alpha\geq10\mathop {\max}\limits_{i}\|X^0_{(i)}\|_2$ guarantees that problem~(\ref{eq:AugModel}) either recovers $\mathcal{X}^0$ exactly or returns an approximate of it nearly as good as the solution of problem~(\ref{eq:ConvexModel}).

There are some directions that the current study can be extended. In this paper, we have focused
on the recovery guarantees of the exact case; it would be meaningful to also analyze the guarantees for
the stable recovery. Second, generalize the linearized Bregman algorithm~\cite{yin2008bregman} for augmented sparse vector and low-rank matrix recovery~\cite{lai2012augmented} to low-rank tensor case. Moreover, from our results, there is a big ``gap'' between the recovery conditions for matrices and tensors, and we need to fill this ``gap'' in future work. In
a broader context, we believe that the current paper could serve as a basis for examining the
augmented trace norm models in tensor recovery.

\section{Acknowledgments}
This work was partly supported by the National Natural Science Foundation of China under grant No. 91120303 and No. 61071181.

\bibliographystyle{IEEEtran}
\bibliography{TITstrings}

\begin{thebibliography}{10}
\providecommand{\url}[1]{#1}
\csname url@samestyle\endcsname
\providecommand{\newblock}{\relax}
\providecommand{\bibinfo}[2]{#2}
\providecommand{\BIBentrySTDinterwordspacing}{\spaceskip=0pt\relax}
\providecommand{\BIBentryALTinterwordstretchfactor}{4}
\providecommand{\BIBentryALTinterwordspacing}{\spaceskip=\fontdimen2\font plus
\BIBentryALTinterwordstretchfactor\fontdimen3\font minus
  \fontdimen4\font\relax}
\providecommand{\BIBforeignlanguage}[2]{{%
\expandafter\ifx\csname l@#1\endcsname\relax
\typeout{** WARNING: IEEEtran.bst: No hyphenation pattern has been}%
\typeout{** loaded for the language `#1'. Using the pattern for}%
\typeout{** the default language instead.}%
\else
\language=\csname l@#1\endcsname
\fi
#2}}
\providecommand{\BIBdecl}{\relax}
\BIBdecl

\bibitem{gandy2011tensor}
S.~Gandy, B.~Recht, and I.~Yamada, ``Tensor completion and low-n-rank tensor
  recovery via convex optimization,'' \emph{Inverse Problems}, vol.~27, p.
  025010, 2011.

\bibitem{tomioka2010estimation}
R.~Tomioka, K.~Hayashi, and H.~Kashima, ``Estimation of low-rank tensors via
  convex optimization,'' \emph{Arxiv preprint arXiv:1010.0789}, 2010.

\bibitem{liu2009tensor}
J.~Liu, P.~Musialski, P.~Wonka, and J.~Ye, ``Tensor completion for estimating
  missing values in visual data,'' in \emph{Computer Vision, 2009 IEEE 12th
  International Conference on}.\hskip 1em plus 0.5em minus 0.4em\relax IEEE,
  2009, pp. 2114--2121.

\bibitem{lai2012augmented}
M.~Lai and W.~Yin, ``Augmented l1 and nuclear-norm models with a globally
  linearly convergent algorithm,'' \emph{Arxiv preprint arXiv:1201.4615}, 2012.

\bibitem{recht2007guaranteed}
B.~Recht, M.~Fazel, and P.~Parrilo, ``Guaranteed minimum-rank solutions of
  linear matrix equations via nuclear norm minimization,'' \emph{Arxiv preprint
  arxiv:0706.4138}, 2007.

\bibitem{mo2011new}
Q.~Mo and S.~Li, ``New bounds on the restricted isometry constant [delta] 2k,''
  \emph{Applied and Computational Harmonic Analysis}, 2011.

\bibitem{kolda2009tensor}
T.~Kolda and B.~Bader, ``Tensor decompositions and applications,'' \emph{SIAM
  review}, vol.~51, no.~3, p. 455, 2009.

\bibitem{de2000multilinear}
L.~De~Lathauwer, B.~De~Moor, and J.~Vandewalle, ``A multilinear singular value
  decomposition,'' \emph{SIAM Journal on Matrix Analysis and Applications},
  vol.~21, no.~4, pp. 1253--1278, 2000.

\bibitem{tomioka2011statistical}
R.~Tomioka, T.~Suzuki, K.~Hayashi, and H.~Kashima, ``Statistical performance of
  convex tensor decomposition,'' \emph{Advances in Neural Information
  Processing Systems (NIPS)}, p. 137, 2011.

\bibitem{donoho2001uncertainty}
D.~Donoho and X.~Huo, ``Uncertainty principles and ideal atomic
  decomposition,'' \emph{Information Theory, IEEE Transactions on}, vol.~47,
  no.~7, pp. 2845--2862, 2001.

\bibitem{gribonval2003sparse}
R.~Gribonval and M.~Nielsen, ``Sparse representations in unions of bases,''
  \emph{Information Theory, IEEE Transactions on}, vol.~49, no.~12, pp.
  3320--3325, 2003.

\bibitem{zhang2005simple}
Y.~Zhang, ``A simple proof for recoverability of l1-minimization: Go over or
  under?'' \emph{Rice University CAAM Technical Report TR05-09}, 2005.

\bibitem{oymak2011simplified}
S.~Oymak, K.~Mohan, M.~Fazel, and B.~Hassibi, ``A simplified approach to
  recovery conditions for low rank matrices,'' in \emph{Information Theory
  Proceedings (ISIT), 2011 IEEE International Symposium on}.\hskip 1em plus
  0.5em minus 0.4em\relax IEEE, 2011, pp. 2318--2322.

\bibitem{candes2005decoding}
E.~Candes and T.~Tao, ``Decoding by linear programming,'' \emph{Information
  Theory, IEEE Transactions on}, vol.~51, no.~12, pp. 4203--4215, 2005.

\bibitem{zhang2008theory}
Y.~Zhang, ``Theory of compressive sensing via l1-minimization: a non-rip
  analysis and extensions,'' \emph{Rice University CAAM Technical Report
  TR08-11}, 2008.

\bibitem{kolda2006multilinear}
T.~Kolda, \emph{Multilinear operators for higher-order decompositions}.\hskip
  1em plus 0.5em minus 0.4em\relax United States. Department of Energy, 2006.

\bibitem{yin2008bregman}
W.~Yin, S.~Osher, D.~Goldfarb, and J.~Darbon, ``Bregman iterative algorithms
  for l1-minimization with applications to compressed sensing,'' \emph{SIAM
  Journal on Imaging Sciences}, vol.~1, no.~1, pp. 143--168, 2008.

\end{thebibliography}

\end{document}